\documentclass{article}
\usepackage{arxiv}

\usepackage{amsthm}

\usepackage{xcolor}  % Needed for colours - load before mdframed

\usepackage{amssymb}
\usepackage{amsmath}
\usepackage{mathdots}

\usepackage{mathtools}

\usepackage{tensor}

\usepackage{urwchancal}
\DeclareFontFamily{OT1}{pzc}{}
\DeclareFontShape{OT1}{pzc}{m}{it}{<-> s * [1.10] pzcmi7t}{}
\DeclareMathAlphabet{\mathpzc}{OT1}{pzc}{m}{it}

\usepackage{graphicx}
\usepackage{ifpdf}
\ifpdf
\usepackage{epstopdf}
\PrependGraphicsExtensions{.svg}
\fi
\newtheorem{theorem}{Theorem}[section]
\newtheorem{lemma}[theorem]{Lemma}

\providecommand{\R}{\mathbb{R}}
\providecommand{\SO}{\mathbf{SO}}

\providecommand{\SE}{\mathbf{SE}}

\providecommand{\so}{\mathfrak{so}}
\providecommand{\se}{\mathfrak{se}}

\providecommand{\Sph}{\mathrm{S}}

\DeclareMathOperator{\tr}{tr}
\providecommand{\trace}[1]{\tr\left(#1\right)}

\DeclareMathOperator{\diag}{diag}

\providecommand{\Lyap}{\mathcal{L}} %% aggregate cost

\providecommand{\td}{\mathrm{d}}

\providecommand{\ddt}{\frac{\td}{\td t}}

\providecommand{\ob}[1]{\overline{#1}} % homogeneous vector
\usepackage{accents}
\makeatletter
\providecommand{\scirc}{%
    \hbox{\fontfamily{\rmdefault}\fontsize{0.4\dimexpr(\f@size pt)}{0}\selectfont{\raisebox{-0.52ex}[0ex][-0.52ex]{$\circ$}}}}
\makeatother
\mathchardef\mhyphen="2D
\usepackage{graphicx}      % include this line if your document contains figures
\usepackage{natbib}        % required for bibliography
\usepackage{hyperref}

\usepackage{verbatim}
\usepackage{ifthen}

\newcommand{\arxivversion}{build the arxiv version} % Comment this line for the nolcos version.

\newenvironment{arxivonly}
{\ifthenelse{\isundefined{\arxivversion}}%
{\expandafter\comment}{}%
}{\ifthenelse{\isundefined{\arxivversion}}%
{\expandafter\endcomment}{}}
\newenvironment{nolcosonly}
{\ifdefined\arxivversion%
\expandafter\comment\fi%
}{\ifdefined\arxivversion%
\expandafter\endcomment\fi}

\newcommand{\papercitation}{% add citation to published paper in text
van Goor, Pieter, Tarek Hamel, and Robert Mahony. "Constructive equivariant observer design for inertial velocity-aided attitude." \textit{IFAC-PapersOnLine 56.1} (2023): 349-354
}

\newcommand{\publicationdetails}
{
\papercitation \\
\copyrightNoticeIFACAccepted{2023} % Use for IFAC papers
}
\newcommand{\publicationversion}
{Author Accepted Version}
\begin{document}

\title{Constructive Equivariant Observer Design for Inertial Velocity-Aided Attitude%
}
\headertitle{Constructive Equivariant Observer Design for Inertial Velocity-Aided Attitude}

\author{
\href{https://orcid.org/0000-0003-4391-7014}{\includegraphics[scale=0.06]{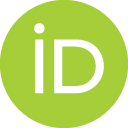}\hspace{1mm}
Pieter van Goor}
\\
    Systems Theory and Robotics Group \\
	Australian National University \\
    ACT, 2601, Australia \\
    \texttt{Pieter.vanGoor@anu.edu.au} \\
	\And	\href{https://orcid.org/0000-0002-7779-1264}{\includegraphics[scale=0.06]{orcid.png}\hspace{1mm}
    Tarek Hamel}
\\
    I3S (University C\^ote d'Azur, CNRS, Sophia Antipolis) \\
    and Insitut Universitaire de France \\
    \texttt{THamel@i3s.unice.fr} \\
	\And	\href{https://orcid.org/0000-0002-7803-2868}{\includegraphics[scale=0.06]{orcid.png}\hspace{1mm}
    Robert Mahony}
\\
    Systems Theory and Robotics Group \\
	Australian National University \\
    ACT, 2601, Australia \\
	\texttt{Robert.Mahony@anu.edu.au} \\
}

\maketitle

\vspace{1cm}

\begin{abstract} % Abstract of not more than 250 words.
Inertial Velocity-Aided Attitude (VAA), the estimation of the velocity and attitude of a vehicle using gyroscope, accelerometer, and inertial-frame velocity (e.g. GPS velocity) measurements,  is an important problem in the control of Remotely Piloted Aerial Systems (RPAS).
Existing solutions provide limited stability guarantees, relying on local linearisation, high gain design, or assuming specific trajectories such as constant acceleration of the vehicle. 
This paper proposes a novel non-linear observer for inertial VAA with almost globally asymptotically and locally exponentially stable error dynamics. 
The approach exploits Lie group symmetries of the system dynamics to construct a globally valid correction term. 
To the authors' knowledge, this construction is the first observer to provide almost global convergence for the inertial VAA problem.
The observer performance is verified in simulation, where it is shown that the estimation error converges to zero even with an extremely poor initial condition.
\end{abstract} % ~150 words

%===============================================================================

%------------------------------------------
\section{Introduction}
%------------------------------------------

Attitude estimation is a fundamental problem in the control of Remotely Piloted Aerial Systems (RPAS).
Many of the most popular approaches from the last 15 years rely on inertial measurement unit (IMU) signals, including gyroscope, accelerometer, and sometimes magnetometer measurements.
A common assumption used in observer design, such as in \cite{2008_mahony_NonlinearComplementaryFilters,2010_martin_DesignImplementationLowcost,2010_pflimlin_ModelingAttitudeControl,2014_hua_ImplementationNonlinearAttitude}, is that the accelerometer signal is dominated by gravity vector.
While this assumption has proven useful in many practical situations, it is unreliable when a RPAS experiences large accelerations such as when sharply changing direction or making turns at high speeds.
A number of authors have incorporated either body-fixed frame or inertial-frame measurements of the RPAS's velocity to overcome this problem.
Estimating the attitude of a RPAS using velocity measurements is referred to as the velocity-aided attitude (VAA) problem.
Velocity aided attitude estimation algorithms typically include an estimator for the vehicle velocity, as well as the primary estimate of the vehicle attitude, that acts as auxiliary state and provides the connection between the velocity measurement and the attitude estimate.

There are two important variations of the VAA problem, depending on whether the velocity of the vehicle is measured in the body-fixed frame (such as provided by an air-data system or a doppler radar), or in the inertial frame (such as provided by a GPS).
Some early solutions to the body-fixed VAA problem were based on linearisation, such as in \citep{2008_bonnabel_SymmetryPreservingObservers,2008_martin_InvariantObserverEarthVelocityAided}
Later solutions tended to be constructive in their design and in some cases provide guarantees of almost-globally asymptotic stability \citep{2013_troni_PreliminaryExperimentalEvaluation,2014_allibert_EstimatingBodyfixedFrame,2016_allibert_VelocityAidedAttitude,2016_hua_StabilityAnalysisVelocityaided,2020_benallegue_VelocityaidedIMUbasedAttitude,2021_wang_NonlinearAttitudeEstimation}.

The solutions to inertial VAA problem solutions are more complex and provide fewer stability guarantees, in general, than the solutions to body-fixed VAA.
In some of the first work on the topic, \cite{2010_hua_AttitudeEstimationAccelerated} provided two observers for the inertial VAA problem.
The first of these observers features semi-global stability by using a high gain design, and the second features almost-global convergence when the vehicle's acceleration is constant by using a virtual $3 \times 3$ matrix in the design.
\cite{2011_roberts_AttitudeEstimationAccelerating} additionally used a magnetometer measurement to provide two observers with semi-global convergence guarantees by using high gain designs.
\cite{2012_grip_NonlinearObserverIntegration} developed an observer for the position, velocity, and attitude of a RPAS using GPS and IMU measurements, and showed the observer error dynamics to be semi-globally exponentially stable by embedding $\SO(3) \hookrightarrow \R^9$ to overcome the topological constraints of the rotation group \citep{2000_bhat_TopologicalObstructionContinuous}.
However, the gains must be tuned carefully to ensure stability, and the attitude estimate provided by the observer is not guaranteed to be continuous.
\cite{2013_dukan_IntegrationFilterAPS} followed a similar approach to create an observer for attitude, position and velocity that takes accelerometer and gyroscope biases into account.
While their experimental results showed that the observer can perform well in practice, the authors did not provide any proof of stability.
A simplified model of quadrotor dynamics was used by \cite{2016_martin_SemiglobalModelbasedState} to develop a semi-globally asymptotically stable observer for the vehicle's pitch, roll, and horizontal velocity, where a gyroscope, an accelerometer, and knowledge of the thrust and rotor drag provide measurements of the velocity and gravity direction directly.
Recently, \cite{2017_hua_RiccatiNonlinearObserver} used a mixed measurement of the vertical component of the inertial frame velocity and the horizontal components of the body-fixed frame velocity to propose a Riccati observer for the inertial VAA problem. 
They proved the error dynamics were locally exponentially stable and used simulation to demonstrate a large domain of attraction 
\citep{2017_hua_RiccatiNonlinearObserver}.

In this paper, we consider the inertial VAA problem without magnetometer measurements.
The design procedure developed by \cite{2021_vangoor_AutonomousErrorConstructive} is applied to propose an observer architecture with synchronous error dynamics.
It is shown that the resulting observer has almost-globally asymptotically stable attitude error dynamics and globally exponentially stable velocity error dynamics.
To the authors' knowledge, this observer is the first that guarantees almost-global convergence, and without assuming constant acceleration.
This contrasts to the previous work, where at most semi-global stability results are known.
Moreover, the auxiliary state used in the proposed observer design is naturally connected to the Lie group structure of the system dynamics viewed through the framework of \citep{2021_vangoor_AutonomousErrorConstructive}.
Simulation results are provided to demonstrate the observer's performance.

%------------------------------------------
\section{Preliminaries}
%------------------------------------------

For an introduction to matrix Lie groups, the authors recommend \cite{2015_hall_LieGroupsLie}.
The special orthogonal group is the Lie group of 3D rotations, defined
\begin{align*}
    \SO(3) := \{
        R \in \R^{3 \times 3} \mid R^\top R = I_3, \; \det(R) = 1
    \}.
\end{align*}
For any vector $\Omega \in \R^3$, define
\begin{align*}
    \Omega^\times
    = \begin{pmatrix}
        0 & -\Omega_3 & \Omega_2 \\
        \Omega_3 & 0 & -\Omega_1 \\
        -\Omega_2 & \Omega_1 & 0
    \end{pmatrix}.
\end{align*}
Then $\Omega^\times v = \Omega \times v$ for any $v \in \R^3$ where $\times$ is the usual vector (cross) product.
The Lie algebra of $\SO(3)$ is defined
\begin{align*}
    \so(3) := \{
        \Omega^\times \in \R^{3 \times 3} \mid \Omega \in \R^3
    \}.
\end{align*}
The projector from $\mathbb{P}_{\so(3)} : \R^{3 \times 3} \to \so(3)$ is defined
\begin{align*}
    \mathbb{P}_{\so(3)}(M) := \frac{1}{2}(M - M^\top).
\end{align*}
For any two vectors $a,b \in \R^3$, one has the following identities:
\begin{align}
    a^\times b &= -b^\times a, &
    (a^\times)^\top &= -a^\times, \notag \\
    a^\times b^\times &= b a^\top - a^\top b I_3, &
    (a \times b)^\times &= b a^\top - a b^\top.
    \label{eq:so3_identities}
\end{align}

The special Euclidean group and its Lie algebra are defined
\begin{align*}
    \SE(3) &:= \left\{
    \begin{pmatrix}
        R & v \\ 0_{1\times 3} & 1
    \end{pmatrix} \in \R^{4 \times 4}
    \;\middle\vert\;
     R \in \SO(3), \; v \in \R^3
    \right\}, \\
    \se(3) &:= \left\{
    \begin{pmatrix}
        \Omega^\times & u \\ 0_{1\times 3} & 0
    \end{pmatrix} \in \R^{4 \times 4}
    \;\middle\vert\;
    \Omega, u \in \R^3
    \right\}.
\end{align*}
An element of $\SE(3)$ may be denoted $X = (R,v)$, where $R \in \SO(3)$ and $v \in \R^3$ for convenience.
Likewise, an element of $\se(3)$ may be denoted $\Delta = (\Omega_\Delta, u_\Delta)$, where $\Omega_\Delta, u_\Delta \in \R^3$.

%------------------------------------------
\section{Problem Description}
%------------------------------------------

Consider a vehicle equipped with an inertial measurement unit (IMU).
Let $\{0\}$ denote the inertial frame and let $\{B\}$ denote the IMU (body) frame.
Then let the attitude and velocity of $\{B\}$ with respect to $\{0\}$ be denoted $R \in \SO(3)$ and $v \in \R^3$, respectively.
The angular velocity $\Omega \in \R^3$ and the linear acceleration $a \in \R^3$ are measured by the IMU.
The dynamical model of $R$ and $v$ considered is
\begin{align}
    \dot{R} &= R \Omega^\times, &
    \dot{v} &= R a + g,
    \label{eq:system_dynamics}
\end{align}
where $g \in \R^3$ is the gravity vector in the inertial frame.
The problem is to design an observer for $R$ using a measurement of the vehicle's velocity in the inertial frame, that is,
\begin{align}
    h(R, v) = v.
    \label{eq:system_measurement}
\end{align}
The velocity state $v$ must also be estimated due to the coupling of the equations of motion \eqref{eq:system_dynamics} and the measurement \eqref{eq:system_measurement}.

%------------------------------------
\section{Observer Design}
%------------------------------------

\subsection{Equivariant Observer Architecture}

%Although $(R,v)$ is not a rigid-body pose, the dynamics \eqref{eq:system_dynamics} can be lifted to the special Euclidean group $\SE(3)$.
Identify the homogeneous matrix
\[
X = \begin{pmatrix} R & v \\ 0_{1 \times 3} & 1 \end{pmatrix} \in \SE(3)
\]
with the state $(R,v)$, noting that although this is not a rigid-body transformation, the symmetry properties of $\SE(3)$ can still be exploited for the VAA state.
We will write $X = (R,v) \in \SE(3)$ to save space in the sequel.
Similarly we write $U = (\Omega, a)$ and $G = (0,g)$ for
\begin{align*}
    U &= \begin{pmatrix} \Omega^\times & a \\ 0_{1 \times 3} & 0 \end{pmatrix} \in \se(3), &
    & \text{and} &
    G &= \begin{pmatrix} 0_{3 \times 3} & g  \\ 0_{1 \times 3} & 0 \end{pmatrix} \in \se(3).
\end{align*}
The lifted dynamics \eqref{eq:system_dynamics} on the Lie group $\SE(3)$ may be written as
\begin{align}
    \dot{X} = X U + G X.
    \label{eq:system_dynamics_se3}
\end{align}
In other words, the Lie group dynamics have both left- and right-invariant components corresponding to body- and spatial- velocities.

Following the observer design procedure described in \citep{2021_vangoor_AutonomousErrorConstructive}, let $\hat{X}, \hat{Z} \in \SE(3)$ be the state estimate and auxiliary state, respectively, and define
\begin{align}
    \dot{\hat{X}} &= \hat{X} U + G \hat{X} + \Delta \hat{X}, &
    \dot{\hat{Z}} &= G \hat{Z} + \hat{Z} \Gamma,
    \label{eq:observer_architecture}
\end{align}
where $\Delta, \Gamma \in \se(3)$ are correction terms that remain to be chosen.

\begin{lemma}
Define an error $\bar{E} := \hat{Z}^{-1} X \hat{X}^{-1} \hat{Z}$.
The system dynamics \eqref{eq:system_dynamics_se3} and the observer internal model \eqref{eq:observer_architecture} are $\ob{E}$-synchronous \citep{2021_vangoor_AutonomousErrorConstructive}; i.e. the dynamics of $\ob{E}$ depend only on the chosen correction terms and $\ddt \ob{E} = 0$ when the correction terms are zero.
\label{lem:synchronous_error}
\end{lemma}

\begin{proof}
Direct computation provides
\begin{align*}
    \dot{\bar{E}}
    &= - \hat{Z}^{-1} \dot{\hat{Z}} \hat{Z}^{-1} X \hat{X}^{-1} \hat{Z}
    + \hat{Z}^{-1} \dot{X} \hat{X}^{-1} \hat{Z}
    \\ &\phantom{=}
    - \hat{Z}^{-1} X \hat{X}^{-1} \dot{\hat{X}} \hat{X}^{-1} \hat{Z}
    + \hat{Z}^{-1} X \hat{X}^{-1} \dot{\hat{Z}}, \\
    %%%%%%%%%%%%%
    &= - \hat{Z}^{-1} G X \hat{X}^{-1} \hat{Z}
    - \Gamma \bar{E}
    + \hat{Z}^{-1} (XU + GX) \hat{X}^{-1} \hat{Z}
    \\ &\phantom{=}
    - \hat{Z}^{-1} X \hat{X}^{-1} (\hat{X} U + G \hat{X} + \Delta \hat{X}) \hat{X}^{-1} \hat{Z}
    \\ &\phantom{=}
    + \hat{Z}^{-1} X \hat{X}^{-1} G \hat{Z}
    + \bar{E} \Gamma, \\
    %%%%%%%%%%%%%
    &= - \Gamma \bar{E}
    - \hat{Z}^{-1} X \hat{X}^{-1} (\Delta \hat{X}) \hat{X}^{-1} \hat{Z}
    + \bar{E} \Gamma, \\
    %%%%%%%%%%%%%
    &= - \Gamma \bar{E} + \bar{E} \Gamma
    - \bar{E} (\hat{Z}^{-1} \Delta \hat{Z}).
\end{align*}
This shows that, indeed, the $\dot{\bar{E}} = 0$ whenever the correction terms $\Delta, \Gamma$ are chosen to be zero, proving the result.
\end{proof}

From \eqref{eq:observer_architecture}, let $\hat{Z} = (R_Z, z) \in \SE(3)$ and $\Gamma = (\Omega_\Gamma, u_\Gamma) \in \se(3)$.
Then,
\begin{align*}
    \dot{R}_Z &= R_Z \Omega_\Gamma^\times, &
    \dot{z} &= g + R_Z u_\Gamma.
\end{align*}
By choosing $R_Z(0) = I_3$ and $\Omega_\Gamma \equiv 0$, it follows that $R_Z \equiv I_3$ and
\begin{align}
    \dot{z} &= g + u_\Gamma.
\label{eq:dot_z}
\end{align}
For the proposed observer design, only the $z$ term is used and attitude component of $R_Z \equiv I_3$ of $\hat{Z}$ is set to the identity. 
Thus, we will only consider $\Gamma = (0, u_\Gamma) \in \se(3)$ in the sequel.
Similarly, we will write \eqref{eq:dot_z} instead of the full $\hat{Z}$ dynamics \eqref{eq:observer_architecture} to emphasise the simplicity of the proposed observer.
In this case $\bar{E} := \hat{Z}^{-1} X \hat{X}^{-1} \hat{Z} = (R_{\bar{E}}, v_{\bar{E}})$ expands to
\begin{equation}
    R_{\bar{E}} = R \hat{R}^\top, \mbox{ and } v_{\bar{E}} = v - z - R \hat{R}^\top (\hat{v} - z).
    \label{eq:error_definitions}
\end{equation}

\subsection{Observer Design}

\begin{theorem}\label{thm:complete_observer_design}
Consider the system dynamics \eqref{eq:system_dynamics}.
Let $\hat{R} \in \SO(3)$, $\hat{v}, z \in \R^3$, and define the observer dynamics \eqref{eq:observer_architecture} written in coordinate form
\begin{subequations}\label{eq:observer_dynamics}
\begin{align}
    \dot{\hat{R}} &= \hat{R} \Omega^\times + \Omega_\Delta^\times \hat{R},&
    \hat{R}(0) &= \hat{R}_0  \label{eq:observerR} \\
    \dot{\hat{v}} &= \hat{R} a + g + \Omega_\Delta^\times \hat{v} + u_\Delta, &
    \hat{v}(0) &= \hat{v}_0, \label{eq:observer_v} \\
    \dot{z} & = g + u_\Gamma, &
    z(0) &= z_0,    \label{eq:z}
\end{align}
\end{subequations}
where $(\hat{R}_0, \hat{v}_0) \in \SE(3)$ and $z_0 \in \R^3$ are arbitrary initial conditions.
Choose the innovation terms $(\Omega_\Delta, u_\Delta, u_\Gamma)$ as follows:
%  and the velocity measurement \eqref{eq:system_measurement}.
\begin{subequations}\label{eq:correction_terms}
\begin{align}
    \Omega_\Delta &= c (\hat{v} - z) \times (v - z),  \\
    u_\Delta &= k(v - \hat{v}) - c ((\hat{v} - z) \times (v - z)) \times z , \\
    u_\Gamma &= k(v - z),
\end{align}
\end{subequations}
with $k, c > 0$ positive gains.
Suppose that the measured velocity \eqref{eq:system_measurement} is bounded and $R a$ is a uniformly continuous and  persistently exciting signal; that is, there exist constants $\mu, \delta > 0$ such that, for all $b \in \Sph^2$ and time $t \geq 0$, there is a $\tau \in [t, t+\delta)$ satisfying
\begin{align}
    \left\vert b^\times (R(\tau) a(\tau)) \right\vert \geq \mu.
    \label{eq:persistent_excitation_ra}
\end{align}
Let $\bar{E} = (R_{\bar{E}}, v_{\bar{E}})$ as in \eqref{eq:error_definitions}.
Then
\begin{enumerate}
    \item The solution $z$ is uniformly continuous $\forall t \geq 0$, and $\bar{E}$ converges to $\bar{E}_s \cup \bar{E}_u$ such that $\bar{E}_s= \{ (I,0) \}$ and
    $$
    \bar{E}_u=\{ (Q, 0) \in \SE(3) \mid \tr(Q) = -1 \}.
    $$
    \item The set $\bar{E}_u$ is the set of unstable equilibria.
    That is, for any point in $\bar{E}_u$ and any neighbourhood ${\cal U}$ of that point, there exists an initial condition in $\mathcal{U}$ such that the error dynamics with that initial condition converge to $\bar{E}_s$.
    \item $\bar{E}_s= \{ (I,0) \}$ is almost-globally asymptotically and locally exponentially stable.
    Moreover, if $\bar{E} \in \bar{E}_s$ then $\hat{R} = R$ and $\hat{v} = v$.
\end{enumerate}

\end{theorem}

%\begin{lemma}\label{lem:basic_convergence}
%Consider the system dynamics \eqref{eq:system_dynamics} and the velocity measurement \eqref{eq:system_measurement}.
%Let $\hat{R} \in \SO(3)$, $\hat{v}, z \in \R^3$ satisfy the observer dynamics \eqref{eq:observer_dynamics} with correction terms \eqref{eq:correction_terms} as in Theorem \ref{thm:complete_observer_design}
%Define the error terms
%\begin{align}
%    R_{\bar{E}} &= R \hat{R}^\top, \notag \\
%    v_{\bar{E}} &= v - z - R \hat{R}^\top \hat{v} + R \hat{R}^\top z.
%    \label{eq:error_definitions}
%\end{align}
%Then $v_{\bar{E}} \to 0$ globally exponentially, and $R_{\bar{E}}^2(v-z) \to (v - z)$ globally asymptotically.
%
%\end{lemma}

% \begin{remark}
% There are four diagonal rotation matrices in $\SO(3)$, the identity, and three diagonal permutations of $(1,-1,-1)$.
% One has $3 \geq \tr(R_{E}) \geq -1$ with equality only for the diagonal rotations.
% \end{remark}

\begin{proof}

\underline{Proof of item 1):} Recalling \eqref{eq:z} along with \eqref{eq:correction_terms} and the fact that $v$ is bounded by assumption, it is straightforward to verify that $z$ is bounded and uniformly continuous.

By direct computation, the dynamics of $R_{\bar{E}}$ are
\begin{align*}
    \dot{R}_{\bar{E}}
    &= (R\Omega^\times) \hat{R}^\top + R (\hat{R} \Omega^\times + \Delta^\times \hat{R}), \\
    &= R\Omega^\times \hat{R}^\top - R \Omega^\times \hat{R}^\top - R \hat{R} \Delta^\times, \\
    &= - R_{\bar{E}} \Omega_\Delta^\times.
\end{align*}
Likewise, the dynamics of $v_{\bar{E}}$ are
\begin{align*}
    \dot{v}_{\bar{E}}
    &= (R a + g) - (g + u_\Gamma) - (- R_{\bar{E}} \Omega_\Delta^\times) \hat{v}
    \\ &\phantom{=}
    - R_{\bar{E}}(\hat{R} a + g + \Omega_\Delta^\times \hat{v} + u_\Delta)
    \\ &\phantom{=}
    + (- R_{\bar{E}} \Omega_\Delta^\times) z + R_{\bar{E}} (g + u_\Gamma), \\
    %%%%%%%%%%%%%%%%%
    &= R a - u_\Gamma + R_{\bar{E}} \Omega_\Delta^\times \hat{v}
    \\ &\phantom{=}
    - R a - R_{\bar{E}} g - R_{\bar{E}} \Omega_\Delta^\times \hat{v} - R_{\bar{E}} u_\Delta
    \\ &\phantom{=}
    - R_{\bar{E}} \Omega_\Delta^\times z + R_{\bar{E}} g + R_{\bar{E}} u_\Gamma, \\
    %%%%%%%%%%%%%%%%%
    &= - u_\Gamma
     - R_{\bar{E}} u_\Delta
    - R_{\bar{E}} \Omega_\Delta^\times z + R_{\bar{E}} u_\Gamma.
\end{align*}
By substituting in the correction terms \eqref{eq:correction_terms},
\begin{align}
    \dot{v}_{\bar{E}}
    &= -u_\Gamma
     - R_{\bar{E}} u_\Delta
    - R_{\bar{E}} \Omega_\Delta^\times z
    + R_{\bar{E}} u_\Gamma, \notag \\
    %%%%%%%%%%%%%%%%%
    &= -k(v - z)
    - R_{\bar{E}} (k(v - \hat{v}) - c ((\hat{v} - z) \times (v - z)) \times z ) \notag
    \\ &\phantom{=}
    - R_{\bar{E}} (c (\hat{v} - z) \times (v - z))^\times z
    + R_{\bar{E}} k(v - z), \notag \\
    %%%%%%%%%%%%%%%%%
    &= -k(v - z)
    - k R_{\bar{E}} (v - \hat{v})
    + k R_{\bar{E}} (v - z), \notag \\
    %%%%%%%%%%%%%%%%%
    &= -k(v - z) + k R_{\bar{E}} (\hat{v} - z), \notag \\
    %%%%%%%%%%%%%%%%%
    &= -k(v - z - R \hat{R}^\top \hat{v} + R \hat{R}^\top z), \notag \\
    %%%%%%%%%%%%%%%%%
    &= -k v_{\bar{E}}, \label{Ve_bar}
\end{align}
which implies that $v_{\bar{E}}$ exponentially converges to zero.

Similarly, by substituting for $\Omega_\Delta$ and using some of the identities \eqref{eq:so3_identities},
\begin{align*}
    \dot{R}_{\bar{E}}
    &= -R_{\bar{E}} \Omega_\Delta^\times, \\
    %%%%%%%%
    &= -c R_{\bar{E}} ((\hat{v} - z) \times (v - z))^\times, \\
    %%%%%%%%
    &= -c R_{\bar{E}} ((v - z)(\hat{v} - z)^\top - (\hat{v} - z) (v - z)^\top), \\
    %%%%%%%%
    &= -c R_{\bar{E}} \left( (v - z)( R_{\bar{E}}^\top(v-z - v_{\bar{E}} ))^\top
    \right. \\ &\phantom{=} \left.
    - ( R_{\bar{E}}^\top(v-z - v_{\bar{E}} )) (v - z)^\top \right), \\
    %%%%%%%%
    &= -c R_{\bar{E}} (v - z)( v-z - v_{\bar{E}} )^\top R_{\bar{E}}
    \\ &\phantom{=}
    + c (v-z - v_{\bar{E}} ) (v - z)^\top.
\end{align*}

Consider the following candidate Lyapunov function,
% \begin{align}
%     \Lyap := \frac{1}{2} \vert R_{\bar{E}} - I_3 \vert^2 + \frac{\alpha}{2} \vert v_{\bar{E}} \vert^2,
%     \label{eq:value_function}
% \end{align}
% with $\alpha > \frac{c}{2k} > 0$.
\begin{align}
    \Lyap &:= \frac{1}{2}\trace{ (\bar{E} - I) A (\bar{E} - I)^\top }, &
    A &:= \begin{pmatrix}
        I_3 & 0_{3 \times 1} \\ 0_{1 \times 3} & \alpha
    \end{pmatrix},
    \label{eq:value_function}
\end{align}
with $\alpha > \frac{c}{2k} > 0$.
One has,
\begin{align*}
    \Lyap
    &= \frac{1}{2} \trace{
    \begin{pmatrix}
        R_{\bar{E}} - I_3 & v_{\bar{E}} \\ 0_{1 \times 3} & 0
    \end{pmatrix}
    \begin{pmatrix}
        I_3 & 0_{3 \times 1} \\ 0_{1 \times 3} & \alpha
    \end{pmatrix}
    \begin{pmatrix}
        R_{\bar{E}}^\top - I_3 & 0_{3 \times 1} \\ v_{\bar{E}}^\top & 0
    \end{pmatrix}
    }, \\
    %%%%%%%%%%%%%%%
    &= \frac{1}{2} \trace{
    \begin{pmatrix}
        (R_{\bar{E}} - I_3)(R_{\bar{E}}^\top - I_3) + \alpha v_{\bar{E}}v_{\bar{E}}^\top & 0_{3 \times 1} \\ 0_{1 \times 3} & 0
    \end{pmatrix}
    }, \\
    %%%%%%%%%%%%%%%
    &= \frac{1}{2} \vert R_{\bar{E}} - I_3 \vert^2 + \frac{\alpha}{2} \vert v_{\bar{E}} \vert^2.
\end{align*}
The dynamics of $\Lyap$ are
\begin{align*}
    \dot{\Lyap}
    &= \left\langle R_{\bar{E}} - I_3, \dot{R}_{\bar{E}} \right\rangle
    + \alpha \left\langle v_{\bar{E}}, \dot{v}_{\bar{E}} \right\rangle, \\
    %%%%%%%%%%%%%%%
    &= \left\langle R_{\bar{E}} - I_3, -c R_{\bar{E}} (v - z)( v-z - v_{\bar{E}} )^\top R_{\bar{E}} \right\rangle
    \\ &\phantom{=}
    + \left\langle R_{\bar{E}} - I_3 , c (v-z - v_{\bar{E}} ) (v - z)^\top \right\rangle
    \\ &\phantom{=}
    - \alpha \left\langle v_{\bar{E}}, k v_{\bar{E}} \right\rangle, \\
    %%%%%%%%%%%%%%%
    &= c \left\langle R_{\bar{E}}^\top R_{\bar{E}}^\top - R_{\bar{E}}^\top,  (v - z)(v-z - v_{\bar{E}} )^\top \right\rangle
    \\ &\phantom{=}
    + c \left\langle R_{\bar{E}}^\top - I_3 , (v - z) (v-z - v_{\bar{E}} )^\top \right\rangle
    - k \alpha \vert v_{\bar{E}} \vert^2, \\
    %%%%%%%%%%%%%%%
    &= c \left\langle R_{\bar{E}}^\top R_{\bar{E}}^\top - I_3, (v - z)( v-z + v_{\bar{E}} )^\top \right\rangle
    - k \alpha \vert v_{\bar{E}} \vert^2, \\
    %%%%%%%%%%%%%%
    &= c ( v-z - v_{\bar{E}} )^\top( R_{\bar{E}}^\top R_{\bar{E}}^\top - I_3) (v - z)
    - k \alpha \vert v_{\bar{E}} \vert^2, \\
    %%%%%%%%%%%%%%
    &= c (v-z)^\top( R_{\bar{E}}^\top R_{\bar{E}}^\top - I_3) (v - z)
    \\ &\phantom{=}
    + c v_{\bar{E}} ^\top( R_{\bar{E}}^\top R_{\bar{E}}^\top - I_3) (v - z)
    - k \alpha \vert v_{\bar{E}} \vert^2, \\
    %%%%%%%%%%%%%%
    &= -\frac{c}{2} \vert (R_{\bar{E}}^2 - I) (v - z) \vert^2 - k \alpha \vert v_{\bar{E}} \vert^2
    \\ &\phantom{=}
    + c v_{\bar{E}} ^\top( R_{\bar{E}}^\top R_{\bar{E}}^\top - I_3) (v - z).
\end{align*}
Then, by using the definition of $\alpha$ and the fact that the Frobenius norm is submultiplicative,
\begin{align}
    \dot{\Lyap} &\leq -\frac{c}{2} \vert (R_{\bar{E}}^2 - I) (v - z) \vert^2 - \frac{c}{2} \vert v_{\bar{E}} \vert^2
    \\ &\phantom{=}
    + c \vert v_{\bar{E}} \vert \vert (R_{\bar{E}}^2 - I) (v - z) \vert
    , \\
    %%%%%%%%%%%%%%
    &= -\frac{c}{2} (\vert (R_{\bar{E}}^2 - I) (v - z) \vert +  \vert v_{\bar{E}} \vert)^2, \label{eq:dot-Lyap}
\end{align}
which is clearly negative semi-definite.

It is easy to see that $\dot{\Lyap}$ is uniformly continuous as it is the product, sum and composition of uniformly continuous functions.
It follows from Barbalat's lemma \cite[Lemma 4.2/4.3]{1991_slotine_AppliedNonlinearControl} that $\dot{\Lyap} \to 0$ and $\Lyap \to \Lyap_{\lim}, \; (\Lyap_{\lim}\leq \Lyap(0))$ a positive constant value.
Combining this with the fact that the equilibrium of the sub-state error $v_{\bar{E}}=0$ \eqref{Ve_bar} is uniformly exponentially stable, one ensures that:
\[\Lyap \to \tr(I-R_{\bar{E}}) \to \Lyap_{\lim}, \mbox{ and }  \ddt \trace{I-R_{\bar{E}}} \to 0. \]

From there, one has
\begin{align*}
    \ddt \trace{R_{\bar{E}}}
    = - \trace{R_{\bar{E}} \Omega_\Delta^\times}
    = - \trace{\mathbb{P}_{\so(3)}(R_{\bar{E}}) \Omega_\Delta^\times}
    \to 0,
\end{align*}
which implies that (i) $\mathbb{P}_{\so(3)}(R_{\bar{E}}) \to 0$ or (ii) $\Omega_\Delta \to 0$.
The first case implies that $R_{\bar{E}} \to R_{\bar{E}}^\top$ and hence $R_{\bar{E}}^2 \to R_{\bar{E}}^\top R_{\bar{E}}=I_3$. The second case directly implies that  $\ddt R_{\bar{E}} \to 0$.
Therefore, in either case one concludes that $\ddt R_{\bar{E}}^2 \to 0$.

Now, using the fact $\ddt (v-z) = R a - k(v-z)$ along with the assumption that $R a$ is persistently exciting, direct application of Lemma \ref{lem:persistent} shows that $(v-z)$ is also persistently exciting.

Since $\dot{\Lyap} \to 0$ in \eqref{eq:dot-Lyap} implies that $R_{\bar{E}}^2(v-z) \to v-z$, it follows that $R_{\bar{E}}^2 \to I_3$ by direct application of Lemma \ref{lem:persistence_of_excitation}.
It follows that $R_{\bar{E}} \to R_{\bar{E}}^\top$, and thus $R_{\bar{E}}$ converges to a symmetric matrix.
From this one concludes that $R_{\bar{E}} \to I_3$, or $R_{\bar{E}} \to U D U^\top$ with $D = \diag(1,-1,-1)$ and $U \in \SO(3)$.
For the latter case, note that for any $Q \in \SO(3)$, $\tr(Q) = -1$ if and only if $Q = U D U^\top$ for some $U \in \SO(3)$.
Therefore, $\bar{E}$ converges asymptotically to $\bar{E}_s$ (for which $\Lyap_{\lim}=0$) or to $\bar{E}_u$ (for which $\Lyap_{\lim}=4$).

\underline{Proof of item 2):}
To see that $\bar{E}=(R_{\bar{E}},0)=(U D U^\top,0)$ is an unstable equilibrium for $D = \diag(1,-1,-1)$ and any $U \in \SO(3)$, recall the definition of $\Lyap$ in \eqref{eq:value_function}.
Let $ \omega = U e_1 \in \R^3$ so that
$$R_{\bar{E}} \omega = U D U^\top U e_1 = U e_1 = \omega,$$
and define $Q(s) = R_{\bar{E}} \exp(s \omega^\times)$.
Define also $\Lyap_u(s) = \Lyap(Q(s), 0)$.
Then,
\begin{align*}
    \Lyap_u(s)
    &= \frac{1}{2} \vert Q(s) - I_3 \vert^2 + \frac{\alpha}{2} \vert 0 \vert^2, \\
    % &= \frac{1}{2} \vert R_{\bar{E}} e^{s \omega^\times } - I_3 \vert^2, \\
    &= \frac{1}{2} \trace{(R_{\bar{E}} e^{s \omega^\times } - I_3)(R_{\bar{E}} e^{s \omega^\times } - I_3)^\top} , \\
    &= \trace{I_3 - R_{\bar{E}} e^{s \omega^\times}}.
\end{align*}
By taking a 2nd order Taylor expansion,
\begin{align*}
    \Lyap_u(s)
    &\approx \trace{I_3 - R_{\bar{E}}(I_3 + s\omega^\times + \frac{s^2}{2} \omega^\times \omega^\times) }, \\
    %%%%%%%%%%%%%%%
    &= \trace{I_3 - R_{\bar{E}}} + s \trace{ R_{\bar{E}} \omega^\times}
    \\ &\phantom{=}
    - \frac{s^2}{2} \trace{ R_{\bar{E}} \omega^\times \omega^\times }.
    \end{align*}
Noting that $\Lyap_u(0) = 4$, $\vert \omega \vert = 1$, and $\trace{ R_{\bar{E}} \omega^\times} = 0$ as $R_{\bar{E}}$ is symmetric, this simplifies to
\begin{align*}
    \Lyap_u(s)&\approx \Lyap_u(0)
    - \frac{s^2}{2} \trace{ R_{\bar{E}} ( \omega \omega^\top - \omega^\top \omega I_3 )  }, \\
    %%%%%%%%%%%%%%%
    &=  4 - \frac{s^2}{2} \trace{ \omega \omega^\top - R_{\bar{E}} }, \\
    %%%%%%%%%%%%%%%
    &=  4 - \frac{s^2}{2} (1 - \trace{ R_{\bar{E}} }), \\
    %%%%%%%%%%%%%%%
    &=  4 - s^2.
\end{align*}
Hence, in any neighbourhood of the equilibrium $(R_{\bar{E}}, 0) = (U D U^\top, 0)$ of $\Lyap$, there exists a perturbation $Q(s)$ such that $\Lyap(Q(s), 0) < \Lyap(R_{\bar{E}} , 0)$.
Therefore the equilibrium $\bar{E}=(U D U^\top, 0)$ is unstable.
Moreover, in any neighbourhood $\mathcal{U}$ of $(U D U^\top, 0)$, there exist initial conditions $(R_{\bar{E}}(0), v_{\bar{E}}(0)) = (Q(s), 0)$ that guarantee $(R_{\bar{E}}, v_{\bar{E}}) \to (I_3, 0)$.
Since the original choice of $\bar{E} \in \bar{E}_u$ was arbitrary, this result holds for all elements of the equilibrium set $\bar{E}_u$.

\underline{Proof of item 3):}
The almost-global asymptotic stability and local exponential stability of $\bar{E}_s= (I_3,0)$ follow from direct application of \citep[Theorem 4.3]{2012_trumpf_AnalysisNonLinearAttitude} along with the persistence of excitation of $(v-z)$ from Lemma \ref{lem:equiv-lemma}.
Finally, supposing that $\bar{E} = (R_{\bar{E}}, v_{\bar{E}}) = (I_3, 0)$, the definitions \eqref{eq:error_definitions} provide
\begin{align*}
    \hat{R} = \hat{R} R^\top R = R_{\bar{E}}^\top R = R,
\end{align*}
and,
\begin{align*}
    \hat{v}
    &= (R \hat{R}^\top)^\top (- v_{\bar{E}} + v - z + R \hat{R}^\top z), \\
    &= R_{\bar{E}}^\top ( v- v_{\bar{E}}) + (I - R_{\bar{E}}^\top) z= v.
\end{align*}
This completes the proof.

\end{proof}

%------------------------------------------
\section{Simulations}
%------------------------------------------

The proposed observer is verified in the following simulation.
Define $R(0) = I_3$, $v(0) = 0$ and let the input signals be
\begin{align*}
    \Omega(t) &= \begin{pmatrix}
        0 & 0 & 1
    \end{pmatrix}^\top, \\
    a(t) &= \begin{pmatrix}
        5 \sin(5t) & 0 & -9.81
    \end{pmatrix}^\top, \\
    g &= \begin{pmatrix}
        0 & 0 & 9.81
    \end{pmatrix}^\top.
\end{align*}
Consider the observer defined in Theorem \ref{thm:complete_observer_design}, and let the initial state be
\begin{align*}
    \hat{R}(0) &= \exp ((\begin{pmatrix}
        2 & -1 & 1.5
    \end{pmatrix}^\top)^\times), \\
    \hat{v}(0) &= \begin{pmatrix}
        3 & -2 & 2
    \end{pmatrix}^\top, \\
    z(0) &= \begin{pmatrix}
        0 & 0 & 0
    \end{pmatrix}^\top.
\end{align*}
The gains are chosen $k = 5,c = 1$.

\begin{figure*}
    \centering
    \includegraphics[width=1.0\textwidth]{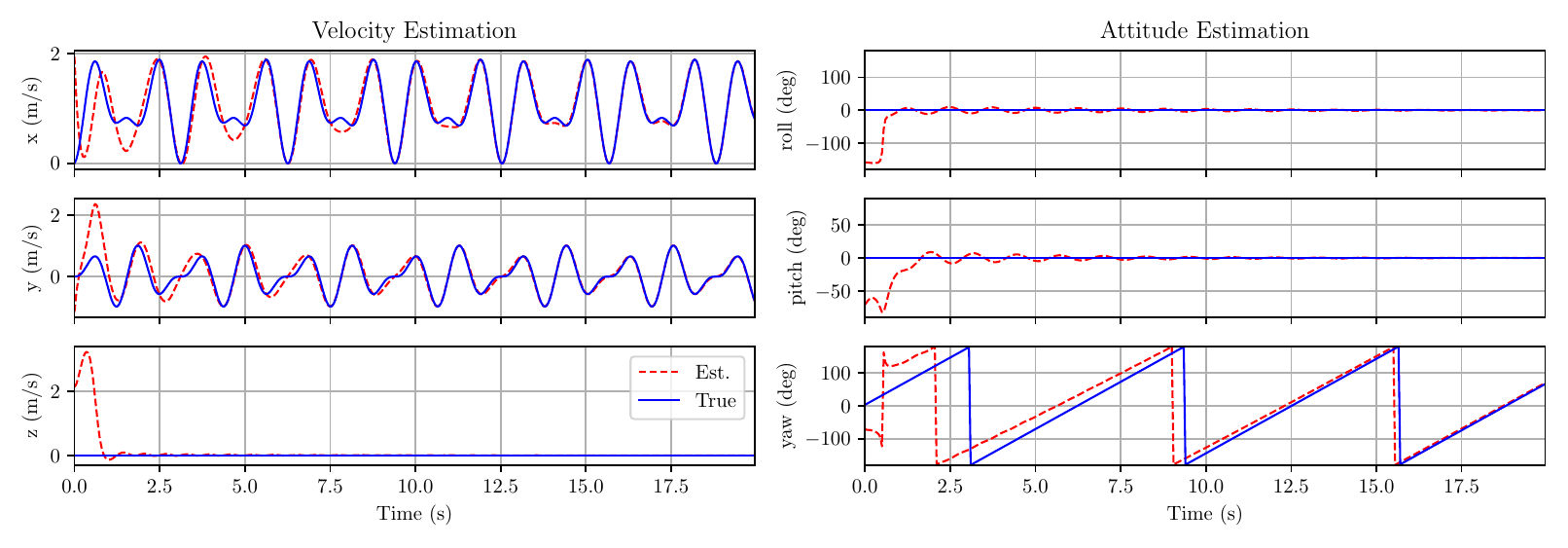}
    \caption{A comparison of the true (solid blue) and estimated (dashed red) velocity and attitude of the example system over time.
    The pitch and roll components of the attitude converge more quickly than the yaw for the observer. 
    }
    \label{fig:estimation}
\end{figure*}

The true system and the observer equations were integrated over a period of 20~s using Euler integration with a time-step of 0.05~s.
Figure \ref{fig:estimation} shows the true and estimated system trajectories over time.
It is clear to see that both the velocity and attitude estimates converge to the true values.
Interestingly, the pitch and roll of the attitude converge more quickly than the yaw, as these are the directions associated with estimating the gravity vector in the body-fixed frame.
Figure \ref{fig:error} shows the attitude error, velocity error, and Lyapunov value of the observer over time.
From these, it can be seen that there is a fast initial convergence followed by a slower second phase of convergence.
Regardless, the Lyapunov value is always clearly decreasing, and the observer is able to estimate the true attitude and velocity despite a large initial attitude error of more than 150$^\circ$.

\begin{figure}
    \centering
    \includegraphics[width=0.75\linewidth]{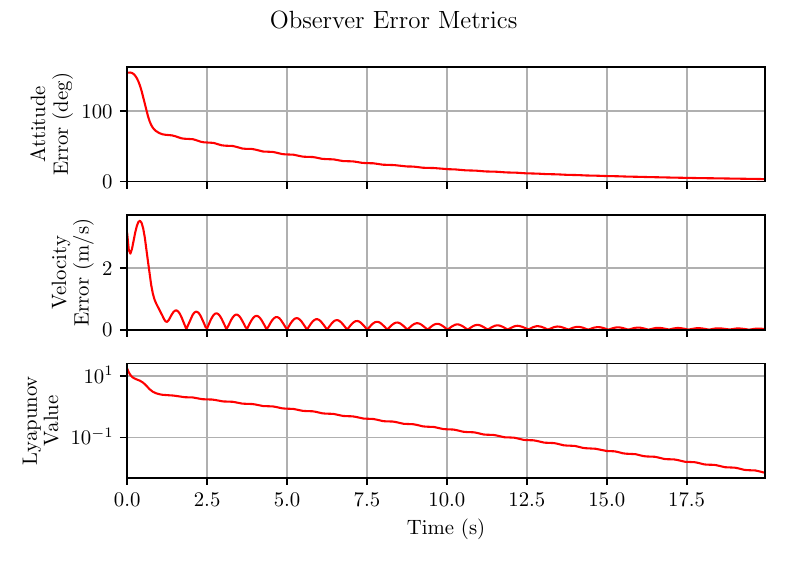}
    \caption{The evolution of the error metrics of the observer state over time.
    The attitude error, velocity error, and Lyapunov value all show a sharp initial decrease, followed by a slower second phase of convergence.
    }
    \label{fig:error}
\end{figure}

%------------------------------------------
\section{Conclusion}
%------------------------------------------

This paper presents a novel constructive observer design for inertial VAA that exploits a modelling of the system dynamics motivated by recent advances in equivariant observer theory \citep{2021_vangoor_AutonomousErrorConstructive}.
This theory provides an observer architecture that is synchronous with the system trajectories, and this paper proposes correction terms that are shown to lead to almost globally asymptotically and locally exponentially stable error dynamics.
To the authors' knowledge, this design is the first to feature such stability properties.
Finally, the provided simulations verify the observer is indeed able to converge from even a large initial error.

%------------------------------------------
\section*{Appendix}
%------------------------------------------

\begin{lemma}\label{lem:persistent}
Let $a(t) \in \R^3$ be uniformly continuous and bounded.
Suppose $a(t)$ is persistently exciting; that is, there exist $\mu, T>0$ such that, for all $t \geq 0$ and $b \in \Sph^2$,
$\vert b^\times a(\tau) \vert \geq \mu$ for some $\tau \in [t, t+T)$.
If $x(t)$ satisfies $\dot{x} = a - k x$ for some fixed $k>0$, then $x(t)$ is also uniformly continuous, bounded and persistently exciting.
\end{lemma}

\begin{arxivonly}
\begin{proof}
The uniform continuity and boundedness of $x$ follow immediately from its dynamics.
Fix $\mu, T$ as in the definition of the persistence of excitation of $a$.
Since $a$ is uniformly continuous, there exists $\delta \in (0, T)$ such that, if $\vert t_2 - t_1 \vert < \delta$, then $\vert a(t_2) - a(t_1) \vert < \mu/2$.

Let $T' := T+\delta$, and let $b \in \Sph^2$ and $t \geq 0$ be arbitrary.
There must exist some $\tau \in [t, t+T)$ such that $\vert b^\times a(\tau) \vert > \mu$.
Moreover, for any $s \in [0,1)$, one has that
\begin{align*}
    \vert b^\times a(\tau + s \delta) \vert &= \vert b^\times a(\tau) + b^\times (a(\tau + s \delta) - a(\tau)) \vert, \\
    &\geq \vert b^\times a(\tau) \vert - \vert b^\times (a(\tau + s \delta) - a(\tau)) \vert, \\
    &\geq \vert b^\times a(\tau) \vert - \vert b \vert \vert a(\tau + s \delta) - a(\tau) \vert, \\
    &> \mu - \mu/2, \\
    &= \mu / 2.
\end{align*}
From there, one ensures that $\vert b^\times a(\tau + s\delta)\vert > 0$ for all $s \in [0,1)$.

Let $\mu' := \frac{\mu}{2(k+2)}$ and suppose $\vert b^\times x(\tau+s\delta) \vert \leq \mu'$ for all $s \in [0,1)$.
Then,
\begin{align*}
    &\vert b^\times x(\tau+\delta) - b^\times x(\tau) \vert \\
    &= \left\vert
        \int_0^1 b^\times \dot{x}(\tau+s \delta) \; \td s
    \right\vert, \\
    %%%%%%
    &= \left\vert
        \int_0^1 b^\times a(\tau + s\delta) - k b^\times x(\tau+s\delta) \; \td s
    \right\vert, \\
    %%%%%%
    &\geq \left\vert
        \int_0^1 b^\times a(\tau + s\delta) \td s
    \right\vert
    - \left\vert
        \int_0^1 k b^\times x(\tau+s\delta) \td s
    \right\vert, \\
    %%%%%%
    &\geq \int_0^1 \left\vert
         b^\times a(\tau + s\delta)
    \right\vert \td s
    - k \int_0^1 \left\vert
        b^\times x(\tau+s\delta)
    \right\vert  \td s, \\
    %%%%%%
    &> \int_0^1 \mu/2 \td s
    - k \int_0^1 \mu' \td s, \\
    %%%%%%
    &= \mu/2 - k \mu', \\
    %%%%%%
    &= \mu/2 - k \mu / (2(k+2)), \\
    %%%%%%
    &= ((k+2)\mu - k \mu) / (2(k+2)), \\
    %%%%%%
    &= 2 \mu / (2(k+2)), \\
    %%%%%%
    &= 2\mu'
\end{align*}
% However, by assumption,
% \begin{align*}
%     \vert b^\top x(\tau+\delta) - b^\top x(\tau) \vert
%     \leq \vert b^\top x(\tau+\delta) \vert + \vert b^\top x(\tau) \vert
%     \leq 2\mu'.
% \end{align*}
% This is a contradiction,
This contradicts the assumption, and therefore there must exist $s \in [0,1)$ such that $\vert b^\times x(\tau+s\delta) \vert > \mu'$.
Put differently, there exists $\tau' = \tau + s\delta \in [0, T')$ such that $\vert b^\times x(\tau') \vert > \mu'$.
Hence $x$ is persistently exciting, as required.

\end{proof}
\end{arxivonly}

\begin{lemma}\label{lem:equiv-lemma}
Assume that $x$ is a uniformly continuous, bounded and persistently exciting signal. Then there exist $\mu', T'>0$ such that, for all $t \geq 0$:
\[ \lambda_2 \left(\int_t^{t+T'}xx^\top \td\tau\right)\geq \mu',\]
with $\lambda_2(S)$ denotes the second largest eigenvalue of a symmetric
matrix $S \in \R^{3\times3}$.
\end{lemma}

\begin{arxivonly}
\begin{proof}
Since $x(t)$ is a persistently exciting signal, then for all $t \geq 0$ and $b \in \Sph^2$ there exists $\tau \in [t,t+T)$ such that $\vert b^\times x(\tau)\vert > \mu$.
This implies that:
\begin{align*}
\vert b^\times x(\tau)\vert^2 &=b^\top(\vert x \vert^2I-xx^\top)b > \mu^2
\end{align*}
Now, since $x$ is uniformly continuous and bounded, it follows that there exists $\mu''$, such that:
\begin{align*}
\int_{t}^{t+T} \vert b^\times x(\tau)\vert^2 \td\tau
&=b^\top\left(\int_{t}^{t+T'}(\vert x \vert^2I-xx^\top)\td\tau \right)b
> \mu''.
\end{align*}
Equivalently, one has that
\begin{align}
    \int_{t}^{t+T}(\vert x \vert^2 I_3 - x x^\top)\td\tau &> \mu'' I_3, \label{eq:projector_pd} \\
    \lambda_{\min} \left(\int_{t}^{t+T'}(\vert x \vert^2I-xx^\top)\td\tau\right) &> \mu'' \label{eq:projector_eigenval}.
\end{align}
% Let $c$ be equal to the trace of left-hand-side of \eqref{eq:projector_pd}; then,
Taking the trace of both sides of \eqref{eq:projector_pd},
\begin{align*}
    \trace{
        \int_{t}^{t+T}(\vert x \vert^2 I_3 - x x^\top) \td\tau
    }
    = 2 \trace{
        \int_{t}^{t+T} x x^\top \td\tau
    }
    > 3 \mu''
\end{align*}
% Taking the trace of both sides of \eqref{eq:projector_pd} yields
% \[2\int_{t}^{t+T'}\vert x \vert^2 \td\tau \geq 3\mu'' \Rightarrow c=\int_{t}^{t+T'}\tr(xx^\top) \td\tau \geq \frac{3}{2}\mu''\]
Let $(\lambda_1,\lambda_2,\lambda_3)$ denote the eigenvalues of $\int_{t}^{t+T'}xx^\top \td\tau$ such that $\lambda_1\geq \lambda_2 \geq \lambda_3$.\
Then $\lambda_1+\lambda_2+\lambda_3 > \frac{3}{2}\mu''$.
Additionally, from \eqref{eq:projector_eigenval}, $\lambda_{\min} = \lambda_2+\lambda_3 > \mu''$.
Since $\lambda_2 \geq \lambda_3$, this ensures that $\lambda_2\geq \frac{1}{2}\mu''$.
\end{proof}
\end{arxivonly}

\begin{lemma} \label{lem:persistence_of_excitation}
Let $Q \in \SO(3)$ such that $\dot{Q}$ converges to zero. Consider a bounded and uniformly continuous persistently exciting signal $x \in \R^3$. That is, there exist constants $\mu, \delta > 0$ such that
\begin{align}
    \lambda_2 \left(\int_t^{t+T'}xx^\top \td\tau\right)\geq \mu',
    \label{eq:persistent_excitation}
\end{align}
according to lemma \ref{lem:equiv-lemma}.
If $(I-Q) x \to 0$, then $Q$ converges to $I$.

\end{lemma}

\begin{arxivonly}
\begin{proof}
Define $\varepsilon, \delta > 0$ so that, for all $t \geq 0$, one has
\begin{align*}
    M(t) := \lambda_2 \left(\int_{t}^{t+\delta} x x^\top \td \tau \right)> \varepsilon.
\end{align*}
Integrating by parts yields
\begin{align*}
    M(t)
    &= \int_{t}^{t+\delta} x x^\top \td \tau, \\
    &= \int_{t}^{t+\delta} Q x x^\top \td \tau, \\
    &= \left[ Q (t+\tau) \int_t^{t+\tau} x(s) x(s)^\top \td s \right]_0^\delta
    \\ &\phantom{==}
    -\int_t^{t+\delta} \left( \frac{\td}{\td \tau} Q(\tau) \right) \left( \int_t^{t+\tau} x(s) x(s)^\top \td s \right) \td \tau, \\
    %%%%%%%%%%%%%%%%
    &\to \left[ Q (t+\tau) \int_t^{t+\tau} x(s) x(s)^\top \td s \right]_0^\delta, \\
    %%%%%%%%%%%%%%%%
    &= Q (t+\delta) \int_t^{t+\delta} x(\tau) x(\tau)^\top \td \tau, \\
    &= Q(t+\delta) M(t)
\end{align*}
Then it follows that $Q M \to M$. Let $(m_1,m_2,m_3)$ be the three orthonormal eigenvectors of $M$ associated with the eigenvalues $(\lambda_1,\lambda_2,\lambda_3)$ with $\lambda_1>\lambda_2 \geq \varepsilon$.
It is straight forward to verify that $(m_1,m_2)$ are also eigenvectors of $Q$ associated to the eigenvalue $1$. Combining this with the fact that $\det(Q)=1$, one concludes that $Q \to I$.
\end{proof}
\end{arxivonly}
%------------

\begin{nolcosonly}
The proofs of these lemmas have been omitted to meet the requirements of the conference page limit.
Please contact the authors for the details of the proofs.
\end{nolcosonly}

%------------------------------------------
\bibliographystyle{plainnat}
\bibliography{NOLCOS_2022_VAA}

\end{document}